\documentclass{article}

\usepackage{lineno}
\usepackage{hyperref}
\usepackage{color}
\usepackage{amsmath}
\usepackage{amsfonts}
\usepackage{amssymb}
\usepackage{geometry}
\usepackage{graphicx}
\usepackage{pgf}
\usepackage{todonotes}
\usepackage[capitalise]{cleveref}
\usepackage{amsopn}
\usepackage{amsthm}
\usepackage[utf8]{inputenc}

\newcommand{\bin}[3]{\operatorname{\mathbf{bin}}({#1}, {#2}, {#3})}
\newcommand{\lbin}[2]{\operatorname{\mathbf{lbin}}({#1}, {#2})}

\newcommand{\vecspace}[2]{\mathbb{Z}_{#1}^{#2}}
\newcommand{\binvecspace}[1]{\vecspace{2}{#1}}
\newcommand{\linearmaps}[2]{\mathcal{L}_{#1}^{#2}}
\newcommand{\surjectivelinearmaps}[2]{\mathcal{LS}_{#1}^{#2}}

\newcommand{\probs}[2]{\operatorname{\mathbf{Pr}}_{{#1}}\left[{#2}\right]}

\newcommand{\expects}[2]{\operatorname{\mathbf{E}}_{{#1}}\left[{#2}\right]}

\newtheorem{theorem}{Theorem}
\newtheorem{proposition}{Proposition}
\theoremstyle{definition}
\newtheorem{definition}{Definition}
\newtheorem*{propositionrep*}{Proposition 3}

\newcommand{\TheTitle}{A Note On the Size of Largest Bins Using Placement With Linear Transformations} 
\newcommand{\TheAuthors}{Martin~Babka}

\title{\TheTitle}
\author{\TheAuthors}

\ifpdf
\hypersetup{
  pdftitle={\TheTitle},
  pdfauthor={\TheAuthors}
}
\fi

\begin{document}

\maketitle

\begin{abstract}
We study the placement of $n$ balls into $n$ bins where balls and bins are represented as two vector spaces over $\binvecspace{}$. The placement is done according to a linear transformation between the two vector spaces.
We analyze the expected size of a largest bin. The only currently known upper bound is $O(\log n \log \log n)$ by Alon et al. and holds for placing $n \log n$ balls into $n$ bins.
We show that this bound can be improved to $O(\log n)$ in the case when $n$ balls are placed into $n$ bins.
We use the same basic technique as Alon et al. but give a tighter analysis for this case.
\end{abstract}

\section{Introduction}

Research of hash function families is nowadays naturally focused on finding fast systems suitable for universal hashing, cuckoo hashing, linear probing, load balancing, etc.
Each application has slightly different requirements on the system.
For example universal hashing~\cite{cw} requires families having small largest bins, for linear probing we have to provide at least a 5-independent family~\cite{linear-probing}.
Additionally the time to compute the hash function should be small.

In this article we are dealing with the size of a largest bin in a balls-and-bins setting.
It is known that if we place $n$ balls into $n$ bins randomly and independently, then with high probability the size of a largest bin is $\Theta(\log n/\log \log n)$.
There are non-trivial hash function families that achieve the sublogarithmic bound such as systems constructed by Siegel~\cite{siegel}, the systems given in~\cite{celisetal}, tabulation hashing~\cite{charhash}, and any $\Theta(\log n/\log \log n)$-independent hash function family. 
The hash function families with high degrees of independence provide asymptotically perfect results for other applications e.g. concentration bounds, Bloom filters, ``two choices'', etc.

Unfortunately the systems with high degrees of independence are inefficient in practice either because of their size and/or speed according to Siegel's lower bound~\cite{siegel}. 
So the research then focused on finding hash function families best fitting the needs of an application. 
There are systems designed to achieve the optimal size of a largest bin for balls-and-bins model that emerged in~\cite{celisetal}.
For cuckoo hashing there are known function families and modifications of the scheme which preserve the expected $O(1)$ operation time such as cuckoo hashing with stash from~\cite{mitzenmacher-cuckoo} and~\cite{dietzfelbinger-cuckoo} without using $\Omega(\log n)$-independent hash function family.
For linear probing it is known that 5-independence is enough to achieve the expected constant probe sequence length~\cite{linear-probing}.

The system of linear transformations between the binary vector spaces forms a natural two-wise independent system of functions.
We show that using this system the size of a largest bin is nearly optimal despite its limited independence.
 Precisely if $n = 2^b$ and $n$ balls, chosen arbitrarily from $\binvecspace{u}$, are placed into $n$ bins using a randomly chosen linear transformation between $\binvecspace{u}$ and $\binvecspace{b}$, then the expected size of a largest bin is $O(\log n)$.
Previously Alon et al \cite{alonetal} showed the bound $O(\log n \log \log n)$ for placement of $n \log n$ balls into $n$ bins.
This bound certainly holds also for placing $n$ balls into $n$ bins.
We improve the previous bound by $\log \log n$ factor when placing $n$ balls into $n$ bins.

We use similar technique as Alon et al. however we use a different parametrization that suits the current setting. 
As a consequence, universal hashing with linear transformations can be implemented so that the amortized running times of the operations match the running times achieved by the balanced trees.

\section{Notation and the setting}
Let $u, b \in \mathbb{N}$, $\vec{a} \in \binvecspace{b}$ and $A$ be a binary matrix of dimension $u \times b$, i.e. $A \in \{0, 1\}^{u \times b}$.
By an \emph{affine linear transformation} from $\binvecspace{u}$ to $\binvecspace{b}$ we understand a mapping $\vec x \mapsto A\vec x + \vec{a}$.
By \emph{linear transformation} from $\binvecspace{u}$ to $\binvecspace{b}$
we understand a mapping $\vec x \mapsto A\vec x$, i.e. an affine transformation with $\vec{a} = \vec{0}$.
Notice that the choice of $\vec{a}$ does not change the bin sizes and thus in our case it is sufficient to analyze the linear transformations only.

By $\linearmaps{u}{b}$ we denote all linear transformations from $\binvecspace{u}$ to $\binvecspace{b}$.
By $\surjectivelinearmaps{u}{b}$ we denote all surjective linear transformations from $\binvecspace{u}$ onto $\binvecspace{b}$.
Let $S \subseteq \binvecspace{u}$ and $T \in \linearmaps{u}{b}$, then by $\lbin{T}{S}$ we denote the size of a largest bin created by $T$ when placing $S$ into $\binvecspace{b}$, i.e. $\lbin{T}{S} = \operatorname{max}_{\vec y \in \binvecspace{b}} |T^{-1}(\vec y) \cap S|$.

When considering probability of an event $E$ or the expected value of a variable $V$ we use the notation $\probs{h \in_U H}{E}$ or $\expects{h \in_U H}{V}$ to indicate that the probability space is formed by the random uniform choice of an object $h$ from a set $H$.

All the logarithms in this article are to the base 2.

\section{Placement of \texorpdfstring{$n$}{n} Balls into \texorpdfstring{$n$}{n} Bins}

In this section we prove Theorem~\ref{theorem-n-to-n} for placement of $n$ balls into $n$ bins using linear transformations.
\begin{theorem}
\label{theorem-n-to-n}
Let $u, b \in \mathbb{N}$, $S \subseteq \binvecspace{u}$ and $|S| \leq 2^b$. Then $\expects{T \in_U \linearmaps{u}{b}}{\lbin{T}{S}} = O(\log |S|)$.
\end{theorem}

We proceeded similarly to \cite{alonetal} and reuse the following propositions from \cite{alonetal}.
\begin{proposition}[{\cite[Theorem~7b, p.~7]{alonetal}}]
\label{proposition-prob-bound}
Let $t, u \in \mathbb{N}$, $t < u$.
Let $S \subset \binvecspace{u}$ such that $\alpha = 1 - \frac{|S|}{2^u}$, $\alpha < 1$.
Then $\probs{T \in_U \surjectivelinearmaps{u}{t}}{T(S) \neq \binvecspace{t}} \leq \alpha^{u - t - \log t + \log \log \frac{1}{\alpha}}$.
\end{proposition}

\begin{proposition}[{\cite[Theorem~{7a}, p.~7]{alonetal}}]
\label{proposition-epsilon}
For each $\epsilon > 0$ there exists $c_\epsilon > 0$ depending solely on $\epsilon$, such that for each $t \in \mathbb{N}, S \subseteq \binvecspace{u}$ satisfying $|S| \geq c_\epsilon t 2^t$ it holds $\probs{T \in_U \linearmaps{u}{t}}{T(S) = \binvecspace{t}} \geq 1 - \epsilon$.
\end{proposition}
Let us note that from the proof in~\cite{alonetal} it follows that $c_\epsilon$ may be chosen as $4\left(\frac{2}{\epsilon}\right)^{\frac{8}{\epsilon}}$.

Following~\cite{alonetal} we define two events needed to estimate the probability of having a  bin of size $\ell$. 
The first event, $E_1$, occurs iff there is a bin of size at least $\ell$.
The second one, $E_2$ is used to upper bound the probability of occurrence of $E_1$.
\begin{definition}[{\cite[Event $E_1$, p.~11]{alonetal}}]
Let $u, b, \ell \in \mathbb{N}$, $T \in \linearmaps{u}{b}$. We put \[E_1(S, T, \ell) \equiv \exists \vec{y} \in \binvecspace{b} \colon |T^{-1}(\vec y) \cap S| \geq \ell.\]
\end{definition}

To define the second event, $E_2$, we decompose the chosen random linear map $T \in \linearmaps{u}{b}$ into $T_0 \in \linearmaps{u}{f}$ and a surjective $T_1 \in \surjectivelinearmaps{f}{b}$ satisfying $T = T_1 \circ T_0$.

\begin{definition}[{\cite[Event $E_2$, p.~11]{alonetal}}]
Let $u, f, b \in \mathbb{N}$, $f \geq b$, $S \subseteq \binvecspace{u}$, $T_0 \in \linearmaps{u}{f}$ and $T_1 \in \surjectivelinearmaps{f}{b}$.
The event $E_2(S, T_0, T_1)$ occurs when $\exists \vec{y} \in \binvecspace{b} \colon T_1^{-1}(\vec y) \subseteq T_0(S)$.
\end{definition}

Refer to \cref{fig-decomposition-general} for the general case of the decomposition and to \cref{fig-decomposition-e2} for the case when $E_2$ occurs. 
Now we show a relation between $E_1$ and $E_2$.

\begin{figure}
	\caption{The decomposition of $T$, general case. Let us note that $F_A = T_1^{-1}(\vec{y})$, $U_A = T^{-1}(\vec y)$ and $S_A = S \cap U_A$.}
	\label{fig-decomposition-general}
\begin{center}
	\begin{tikzpicture}
		\draw[->] (4,2) -- (5,2) node[left=0.5cm,below] {$T_0$};
		\draw[green] (2,2) circle (2cm) node[left=1cm,below=1cm] {$\binvecspace{u}$};
		\draw[green] (6.75,2) circle (1.75cm) node[left=0.75cm,below=0.75cm] {$\binvecspace{f}$};
		\draw[->] (8.5,2) -- (9.5,2) node[left=0.5cm,below] {$T_1$};
		\draw[green] (11,2) circle (1.5cm) node[left=0.6cm,below=0.6cm] {$\binvecspace{b}$};
		
		\draw[blue] (10.5,2.75) circle (0.05cm) [fill=black] node[anchor=west] {$\vec{y}$};
		\draw[->,red] (10.5,2.75) -- (7.5,2.75) node[left=-1.7cm,above] {$T_1^{-1}$};
		\draw[dashed,red] (6.75,2.75) circle (0.75cm)  node[] {$F_A$};
		\draw[dotted,blue] (6.2,1.95) circle (0.75cm)  node[left=0cm,below=0.1cm] {$T_0(S)$};

		\draw[dashed,red] (2.2,2.75) circle (1cm)  node[left=-0.35cm,below=-0.5cm] {$U_A$};
		\draw[dotted,blue] (1.4,2.1) circle (1.2cm)  node[left=0.35cm,below=0.1cm] {$S$};
		
		\draw[magenta] (1.7,2.5) node[] {$S_A$};
		
		\draw[->] (10.5,2.75) arc (0:152:4.2cm and 2cm) node[midway,below=0.0cm] {$T^{-1}$};
		
		\draw[->,red] (3.2,2.75) -- (6,2.75) node[left=1.4cm,above] {$T_0$};
	\end{tikzpicture}
\end{center}
\end{figure}
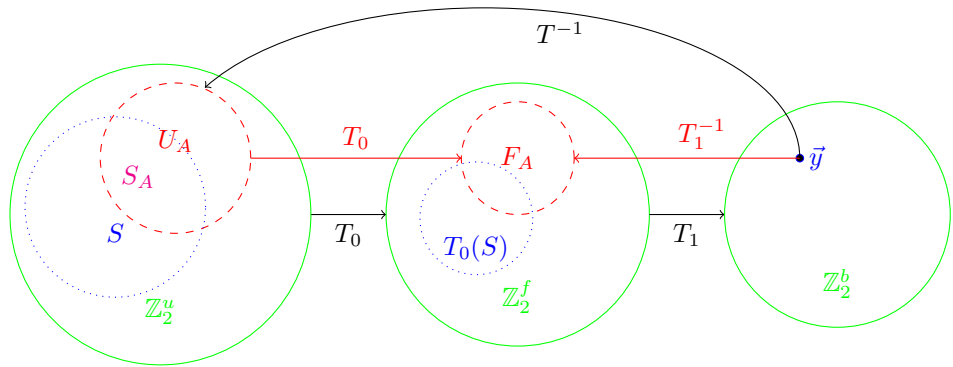

\begin{proposition}[{\cite[Proposition 3.2, p.~11]{alonetal}}]
\label{proposition-e1-e2}
For each $\epsilon > 0$ there is $c_\epsilon > 0$ such that for each $u, f, b, \ell \in \mathbb{N}$ satisfying $u \geq f \geq b$, $\ell \geq c_\epsilon (f - b)2^{f-b}$ and for arbitrary $S \subseteq \binvecspace{u}$, it holds that
$
\probs{T \in_U \linearmaps{u}{b}}{E_1(S, T, \ell)} \leq \frac{1}{1 - \epsilon}\probs{T_0 \in_U \linearmaps{u}{f}, T_1 \in_U \surjectivelinearmaps{f}{b}}{E_2(S, T_0, T_1)}.
$
In addition the value $c_\epsilon$ can be chosen according to \cref{proposition-epsilon} and depends only on $\epsilon$.
\end{proposition}

For completeness we provide a proof of \cref{proposition-e1-e2} in the appendix.

Now we estimate the probability of $E_2$.
Our \cref{proposition-e2-bound} is a slight restatement of Proposition~3.1 from \cite{alonetal}.
It is similar to Proposition~3.1 in \cite{alonetal} but gives a slightly better bound.
The proof is similar.

\begin{proposition}
\label{proposition-e2-bound}
Let $u, f, b \in \mathbb{N}$ such that $u \geq f > b$.
If $S \subseteq \binvecspace{u}$, $|S| = 2^b$, $f > b$ and $\mu = \frac{2^b}{2^f}$, then
$
\probs{T_0 \in_U \linearmaps{u}{f}, T_1 \in_U \surjectivelinearmaps{f}{b}}{E_2(S, T_0, T_1)} \leq \mu ^ {-\log b - \log \mu + \log \log \mu^{-1}}.
$
\end{proposition}
\begin{proof}
Observe that $\exists \vec{y} \in \binvecspace{b} \colon T_1^{-1}(\vec{y}) \subseteq T_0(S)$ is equivalent to $\exists \vec{y} \in \binvecspace{b} \colon \vec{y} \not\in T_1(\binvecspace{f} \setminus T_0(S))$.
Hence $E_2(S, T_0, T_1)$ is equivalent to $T_1(\binvecspace{f} \setminus T_0(S)) \neq \binvecspace{b}$. 
Refer to \cref{fig-decomposition-e2} for more details of the situation when $E_2(S, T_0, T_1)$ occurs.

\begin{figure}[h]
	\caption{Decomposition of $T$ when event $E_2(S, T_0, T_1)$ occurs, i.e. $F_A \subseteq T_0(S)$.}
	\label{fig-decomposition-e2}
\begin{center}
	\begin{tikzpicture}
		\draw[->] (4,2) -- (5,2) node[left=0.5cm,below] {$T_0$};
		\draw[green] (2,2) circle (2cm) node[left=1cm,below=1cm] {$\binvecspace{u}$};
		\draw[green] (6.75,2) circle (1.75cm) node[left=0.75cm,below=0.75cm] {$\binvecspace{f}$};
		\draw[->] (8.5,2) -- (9.5,2) node[left=0.5cm,below] {$T_1$};
		\draw[green] (11,2) circle (1.5cm) node[left=0.6cm,below=0.6cm] {$\binvecspace{b}$};
		
		\draw[blue] (10.5,2.75) circle (0.05cm) [fill=black] node[anchor=west] {$\vec{y}$};
		\draw[->,red] (10.5,2.75) -- (7.5,2.75) node[left=-1.7cm,above] {$T_1^{-1}$};
		\draw[dashed,red] (6.75,2.75) circle (0.75cm) node[] {$F_A$};
		\draw[dotted,blue] (6.8,2.35) circle (1.2cm)  node[left=0cm,below=0.5cm] {$T_0(S)$};

		\draw[dashed,red] (2.2,2.75) circle (1cm)  node[left=-0.35cm,below=-0.5cm] {$U_A$};
		\draw[dotted,blue] (1.4,2.1) circle (1.2cm)  node[left=0.35cm,below=0.1cm] {$S$};

		\draw[->] (10.5,2.75) arc (0:152:4.2cm and 2cm) node[midway,below=0.0cm] {$T^{-1}$};
		
		\draw[magenta] (1.7,2.5) node[] {$S_A$};
				
		\draw[->,red] (3.2,2.75) -- (6,2.75) node[left=1.4cm,above] {$T_0$};
	\end{tikzpicture}
\end{center}
\end{figure}
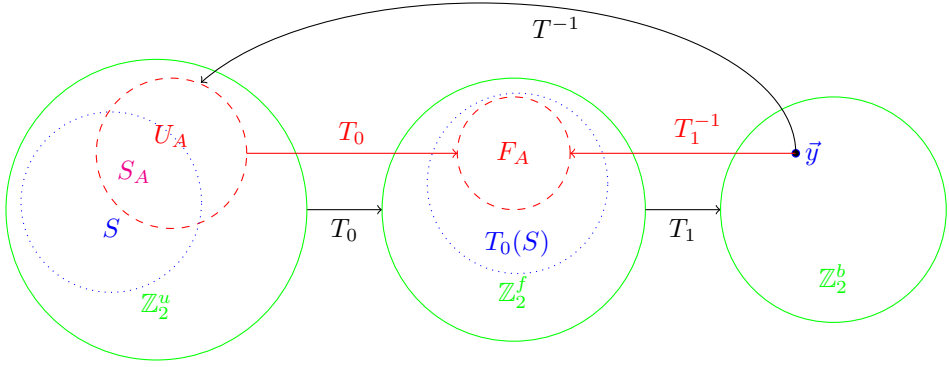

We prove the estimate for arbitrary fixed $T_0$ and uniform choice of $T_1$.
From \cref{proposition-prob-bound} it follows that $\probs{T_1\in_U \surjectivelinearmaps{f}{b}}{T_1(\binvecspace{f} \setminus T_0(S)) \neq \binvecspace{b}} \leq \alpha ^ {f - b - \log b + \log \log \alpha^{-1}}$ where $\alpha = 1 - \frac{|\binvecspace{f} \setminus T_0(S)|}{|\binvecspace{f}|} = \frac{|T_0(S)|}{2^f} \leq \frac{|S|}{2^f} = \mu = 2^{b - f} < 1$.
Since the function $\alpha ^ {f - b - \log b + \log \log \alpha^{-1}}$ is increasing w.r.t. $\alpha$ in $(0, 1)$ we get that
$
\probs{T_1 \in_U \surjectivelinearmaps{f}{b}}{E_2(S, T_0, T_1)} \leq \mu ^ {-\log \mu - \log b + \log \log \mu^{-1}}.
$
\end{proof}

The following theorem gives an upper bound for the tail distribution of the random variable $\lbin{T}{S}$.
The theorem is similar to Corollary~3.3 from \cite{alonetal}, however the stated estimate is slightly different because it is adapted to our setting.
The substantial difference between them is that we obtain non-trivial estimates for the logarithmic size of a largest bin whereas in \cite{alonetal} they get them for super-logarithmic sizes.
The theorem in turn implies \cref{theorem-n-to-n}.

\begin{theorem}
\label{theorem-prob-distribution-bound}
For each $\epsilon > 0$ there exists $c_\epsilon > 0$ such that for each $u, b \in \mathbb{N}, u \geq b$, $r \geq 	 4$ it holds that
\[
\probs{T \in_U \linearmaps{u}{b}}{\lbin{T}{S} \geq 2 c_\epsilon r} \leq \frac{1}{1 - \epsilon}\left(\frac{\log r}{r}\right)^{-\log b - \log \frac{\log r}{r} + \log \log \frac{r}{\log r}}.
\]
Moreover $c_\epsilon$ depends solely on $\epsilon$ and may be chosen according to \cref{proposition-epsilon}.
\end{theorem}
\begin{proof}
Let $\epsilon, u, b, r$ be given so that they meet the requirements of the theorem.
We put $f = \lfloor b + \log r - \log \log r + 1 \rfloor$ and $\ell = \left \lceil 2c_\epsilon r \right\rceil$ where $c_\epsilon$ comes from \cref{proposition-epsilon}.
Recall that $\lbin{T}{S} \geq \ell$ is equivalent to the occurrence of event $E_1(S, T, \ell)$.

\cref{proposition-e1-e2} implies that $\probs{T \in_U \linearmaps{u}{b}}{E_1(S, T, \ell)} \leq \frac{1}{1 - \epsilon}\probs{T_0 \in_U \linearmaps{u}{f}, T_1 \in_U \surjectivelinearmaps{f}{b}}{E_2(S, T_0, T_1)}$.
We have to verify that $\ell \geq c_\epsilon (f - b)2^{f - b}$.
From the requirement $r \geq 4$ it follows that $c_\epsilon(f - b)2^{f - b} \leq c_\epsilon(\log r - \log \log r + 1)2^{\log r - \log \log r + 1} \leq \frac{2c_\epsilon r(\log r - \log \log r + 1)}{\log r} \leq 2c_\epsilon r \leq \ell$.

To bound the probability of $E_2(S, T_0, T_1)$ we use \cref{proposition-e2-bound}.
Observe that the choice of $f$ from the beginning of the proof satisfies $f > b$.
This also means that $\surjectivelinearmaps{f}{b}$ is nonempty. 
We put $\mu = 2^{b - f}$.
Since $\mu \leq 2^{-\log r + \log \log r} = \frac{\log r}{r}$ and the function $g(x) := x ^ {- \log b + \log x^{-1} + \log \log x^{-1}}$ is increasing in $(0, 1)$, from \cref{proposition-e2-bound} it follows that $\probs{T_0 \in_U \linearmaps{u}{f}, T_1 \in_U \surjectivelinearmaps{f}{b}}{E_2(S, T_0, T_1)} \leq g(\mu) \leq g\left(\frac{\log r}{r}\right)$.
\end{proof}

Now we show the proof of the main theorem.

\begin{proof}[Proof of \cref{theorem-n-to-n}]
We show the theorem for $|S| = 2^b$.
If $|S| < 2^b$, the theorem follows from the proved case.
Put $n = 2^b = |S|$.
We split $\sum_{\ell = 1}^{n} \probs{T\in_U\linearmaps{u}{b}}{\lbin{T}{S} \geq \ell}$ into two sums according to $\ell$ being lower or greater than $8c_\epsilon \log n$.
We show that in the second case the probability of $\lbin{T}{S} \geq \ell$ is $O(\ell^{-3/2})$.

First we fix $\epsilon \in (0, 1)$ arbitrarily, assume that $\ell \geq 8c_\epsilon n$ and choose $r$ so that $\ell = 2 c_\epsilon r$.
Hence $r \geq 4\log n$.
We claim that if $n$ is large enough, then the estimate obtained by \cref{theorem-prob-distribution-bound} is below $\frac{r^{-1.5}}{1-\epsilon}$.
To prove this we bound the exponent of the estimate from below as follows.
\begin{align*}
-\log b - \log \log r + \log r + \log (\log r - \log \log r) 
	& \geq -\log \log r + 2 + \log \left(\frac{3\log r}{4}\right) \\
	& = \log(3) > \frac{3}{2}.
\end{align*}
Hence when $n$ is large enough we get that $\left(\frac{\log r}{r}\right)^{\log 3} < r^{-3/2}$ and
\begin{align*}
\sum_{\ell = 8c_\epsilon \log n + 1}^{n} \probs{T\in_U\linearmaps{u}{b}}{\lbin{T}{S} \geq \ell} 
	& \leq \int_{8c_\epsilon \log n}^{n} \probs{T\in_U\linearmaps{u}{b}}{\lbin{T}{S} \geq \ell} d\ell \\
	& = 2c_\epsilon \int_{4 \log n}^{n/2c_\epsilon} \probs{T\in_U\linearmaps{u}{b}}{\lbin{T}{S} \geq 2c_\epsilon r} dr \\
	& \leq \frac{2c_\epsilon}{1-\epsilon} \int_{1}^{\infty} r^{-1.5} dr = O\left(\frac{c_\epsilon}{1-\epsilon}\right).
\end{align*}

The whole sum may be estimated as $
\sum_{\ell = 1}^{n} \probs{T\in_U\linearmaps{u}{b}}{\lbin{T}{S} \geq \ell} \leq 2c_\epsilon \left(4 \log n + \frac{O(1)}{1-\epsilon}\right).$
\end{proof}

\section{The special case when balls form a vector subspace}

Let us note that when $S$ is a subspace of the universe, then the expected size of the largest bin is constant.

\begin{theorem}
Let $b, u \in \mathbb{N}$ and $S$ be a subspace of $\binvecspace{u}$ of dimension $b$. Then \[ \expects{T \in_U \linearmaps{u}{b}}{\lbin{T}{S}} = O(1).\]
\end{theorem}
\begin{proof}
We first observe that the non-empty bins have a simple structure -- all of them are formed by elements which are affine subspaces of the universe.
This in turns means that all the non-empty bins have the same size.
Since the bin containing $\vec{0}$ in $\binvecspace{b}$ is always non-empty and has a constant expected size, the theorem follows.

Assume that $T \in \linearmaps{u}{b}$ is fixed.
Let $K = S \cap \operatorname{Ker}(T)$.
If $T(\vec{v}) = \vec{y}$ for some $\vec{v} \in S$, then $T^{-1}(\vec y) \cap S = \vec v + K$.
Hence for each $\vec y \in \mathbb{Z}_2^b$ it holds that $|T^{-1}(\vec y) \cap S| = 0$ or $|T^{-1}(\vec y) \cap S| = |K|$.
By $\bin{T}{S}{\vec{y}}$ we denote $|T^{-1}(\vec{y}) \cap S|$ and it holds that $|K| = \bin{T}{S}{\vec 0}$.
From this it follows that 
$
\expects{T\in_U \linearmaps{u}{b}}{\lbin{T}{S}} = \expects{T\in_U \linearmaps{u}{b}}{\bin{T}{S}{\vec 0}} = O(1).
$

\end{proof}

\section{Acknowledgment}

We would like to thank V\'aclav Koubek and Michal Kouck\'y for advices, consultations and time spent verifying this note.

\bibliographystyle{plain}
\bibliography{lhf}

\begin{thebibliography}{1}

\bibitem{alonetal}
Noga Alon, Martin Dietzfelbinger, Peter~Bro Miltersen, Erez Petrank, and
  G\'{a}bor Tardos.
\newblock Linear hash functions.
\newblock {\em J. ACM}, 46(5):667--683, September 1999.

\bibitem{dietzfelbinger-cuckoo}
Martin Aum{\"u}ller, Martin Dietzfelbinger, and Philipp Woelfel.
\newblock Explicit and efficient hash families suffice for cuckoo hashing with
  a stash.
\newblock {\em Algorithmica}, 70(3):428--456, 2014.

\bibitem{cw}
J.Lawrence Carter and Mark~N. Wegman.
\newblock Universal classes of hash functions.
\newblock {\em Journal of Computer and System Sciences}, 18(2):143 -- 154,
  1979.

\bibitem{celisetal}
L.~Elisa Celis, Omer Reingold, Gil Segev, and Udi Wieder.
\newblock Balls and bins: Smaller hash families and faster evaluation.
\newblock {\em SIAM Journal on Computing}, 42(3):1030--1050, 2013.

\bibitem{mitzenmacher-cuckoo}
Adam Kirsch, Michael Mitzenmacher, and Udi Wieder.
\newblock More robust hashing: Cuckoo hashing with a stash.
\newblock {\em SIAM Journal on Computing}, 39(4):1543--1561, 2010.

\bibitem{linear-probing}
Anna Pagh, Rasmus Pagh, and Milan Ružić.
\newblock Linear probing with constant independence.
\newblock {\em SIAM Journal on Computing}, 39(3):1107--1120, 2009.

\bibitem{charhash}
Mihai P\v{a}tra\c{s}cu and Mikkel Thorup.
\newblock The power of simple tabulation hashing.
\newblock {\em J. ACM}, 59(3):14:1--14:50, June 2012.

\bibitem{siegel}
Alan Siegel.
\newblock On universal classes of extremely random constant-time hash
  functions.
\newblock {\em SIAM Journal on Computing}, 33(3):505--543, 2004.

\end{thebibliography}

\begin{appendix}
\section{Proof of \cref{proposition-e1-e2}}
We give the full proof of \cref{proposition-e1-e2} along with the necessary claims. 
\begin{proposition}
\label{claim-dstr-factor}
Let $T_1 \in \surjectivelinearmaps{f}{b}$ be fixed. Then the uniform choice of $T_0 \in \linearmaps{u}{f}$ yields the uniform choice of $T \in \linearmaps{u}{b}$ where $T = T_1 \circ T_0 $.
\end{proposition}
\begin{proof}
The proof of the claim may be found in \cite{alonetal} in the proof of Theorem~7b.
Let $\vec{e}_1, \dots, \vec{e}_u$ be a basis of $\binvecspace{u}$.
Recall that the uniform choice of $T \in \linearmaps{u}{b}$ is equivalent to random and independent choice of $T(\vec{e}_i) \in \binvecspace{b}$ for $i \in \{1, \dots, u\}$.
Since $T_1$ is onto, for each $\vec{y} \in \binvecspace{b}$ we have that $|T_1^{-1}(\vec{y})| = 2^{f - b}$.
Hence the uniform independent choice of values $T_0(\vec{e}_i) \in \binvecspace{f}$ yields uniform independent choice of values $T(\vec{e}_i) = T_1(T_0(\vec{e}_i)) \in \binvecspace{b}$ for $i \in \{1, \dots, u\}$.
\end{proof}

\begin{proposition}[{\cite[Proposition~3.4, p.~13]{alonetal}}]
\label{proposition-affine-model}
Let $u, f, b \in \mathbb{N}$, such that $u \geq f \geq b$.
For a fixed $T \in \linearmaps{u}{b}$ and $T_1 \in \surjectivelinearmaps{f}{b}$ there is a bijection between $\{ T_0 \in \linearmaps{u}{f} \mid T = T_1 \circ T_0 \}$ and linear maps from $\operatorname{Ker}(T)$ to $\operatorname{Ker}(T_1)$.
\end{proposition}
\begin{proof}
We show that when $T$ and $T_1$ are fixed, then each restriction of $T_0$ to $\operatorname{Ker}(T)$ can be uniquely extended to $\binvecspace{u}$.
Thus the bijection is defined as $T_0\mathord{\upharpoonright}\operatorname{Ker}(T) = T_K$ where $T_K$ is a linear map from $\operatorname{Ker}(T)$ to $\operatorname{Ker}(T_1)$.

Let $B$ be a orthogonal basis of $\operatorname{Ker}(T)$ and $B^E$ be an orthogonal extension of $B$ to $\binvecspace{u}$.
Similarly let $B_1$ be a orthogonal basis of $\operatorname{Ker}(T_1)$ and $B_1^E$ be an orthogonal extension of $B_1$ to $\binvecspace{f}$.

Let $\vec{x} \in \binvecspace{u}$.
There exists a unique decomposition of $\vec{x}$ into two vectors $\vec{x}_K \in \operatorname{Span}(B)$ and $\vec{x}_C \in \operatorname{Span}(B^E \setminus B)$ such that $\vec{x} = \vec{x}_K + \vec{x}_C$.
Analogically there is a unique vector $\vec{q}_{\vec{x}} \in \operatorname{Span}(B_1^E \setminus B_1)$ satisfying that $T_1(\vec{q}_{\vec{x}}) = T(\vec{x}) = T(\vec{x}_C)$.
We put $T_0(\vec{x}) = \vec{q}_{\vec{x}} + T_K(\vec{x}_K)$.
\end{proof}

\begin{propositionrep*} \emph{({\cite[Proposition 3.2, p.~11]{alonetal}}).}
For each $\epsilon > 0$ there is $c_\epsilon > 0$ such that for each $u, f, b, \ell \in \mathbb{N}$ satisfying $u \geq f \geq b$, $\ell \geq c_\epsilon (f - b)2^{f-b}$ and for arbitrary $S \subseteq \binvecspace{u}$, it holds that
$
\probs{T \in_U \linearmaps{u}{b}}{E_1(S, T, \ell)} \leq \frac{1}{1 - \epsilon}\probs{T_0 \in_U \linearmaps{u}{f}, T_1 \in_U \surjectivelinearmaps{f}{b}}{E_2(S, T_0, T_1)}.
$
In addition the value $c_\epsilon$ can be chosen according to \cref{proposition-epsilon} and depends only on $\epsilon$.
\end{propositionrep*}
\begin{proof}[Proof of \cref{proposition-e1-e2}]
Fix $\epsilon \in (0, 1)$.
First we show that $1 - \epsilon \leq \probs{T_0, T_1}{E_2 | E_1}$.
Assume that $E_1(S, T, \ell)$ occurs, i.e. there is $\vec{y} \in \binvecspace{b}$ such that $|T^{-1}(\vec{y}) \cap S| \geq \ell$.
Put $U_A = T^{-1}(\vec{y})$, $S_A = U_A \cap S$ and $F_A = T_1^{-1}(\vec{y})$.
If $T_0(S_A) = F_A$, then $T_0(S) \supseteq T_0(S_A) = F_A = T_1^{-1}(\vec{y})$ and by definition $E_2(S, T_0, T_1)$ occurs.
See \cref{fig-decomposition-e2} for a better picture of the situation when $E_2(S, T_0, T_1)$ occurs.
Thus it is sufficient to estimate $\probs{T_0, T_1}{T_0(S_A) = F_A | E_1(S, T, \ell)}$.
To do so we further assume that $T_1$, $T$ are fixed, $T = T_1 \circ T_0$ and $E_1(S, T, \ell)$ occurs.

Since $T_1$ is onto, it holds that $|F_A| = 2^{f-b}$.
Also notice that $U_A$ and $F_A$ are affine subspaces of $\binvecspace{u}$ and $\binvecspace{f}$ and $|S_A| \geq \ell \geq c_\epsilon (f - b)2^{f-b}$.
Let $T_A$ be an affine linear map from $U_A$ to $F_A$.
From \cref{proposition-epsilon} used for $T_A$, $U_A$, $S_A$, $F_A$ we get that $\probs{T_A}{T_A(S_A) = F_A | E_1} \geq 1 - \epsilon$.
Notice that \cref{proposition-epsilon} may be used for affine linear transformations as well.
Since the previous estimate holds for arbitrary fixed $T$ and $T_1$, it holds for the uniform choice of the two transformations. Thus $1 - \epsilon \leq \probs{T_A}{T_A(S_A) = F_A | E_1} = \probs{T_A, T, T_1}{T_A(S_A) = F_A | E_1}$.
From previous and \cref{proposition-affine-model} we get that $1 - \epsilon \leq \probs{T_0, T_1}{E_2 | E_1}$.
From the previous inequality and \cref{claim-dstr-factor} it follows that $\probs{T}{E_1} = \probs{T_0, T_1}{E_1} \leq \frac{1}{1-\epsilon}\probs{T_0, T_1}{E_2}$.
\end{proof}
\end{appendix}
\end{document}